\documentclass[10pt,conference]{IEEEtran}
\usepackage[cmex10]{amsmath}
\usepackage{mathtools}
\usepackage{amssymb}
\usepackage{amsthm}
\usepackage{epsfig}
\usepackage{color}
\usepackage{algorithm}
\usepackage{algorithmic}
\usepackage{multirow}
\usepackage{array}
\newtheorem{theorem}{Theorem}

\newtheorem{corollary}{Corollary}
\newtheorem{lemma}{Lemma}
\usepackage{epsfig}
\usepackage{color}

\begin{document}
\title{Optimal Data Attacks on Power Grids:\\ Leveraging Detection \& Measurement Jamming}
\author{\IEEEauthorblockN{Deepjyoti Deka,~ Ross Baldick ~and~ Sriram Vishwanath}
\IEEEauthorblockA{Department of Electrical \& Computer Engineering, The University of Texas at Austin\\
Email: deepjyotideka@utexas.edu, baldick@ece.utexas.edu, sriram@ece.utexas.edu }}

\maketitle
\begin{abstract}
Meter measurements in the power grid are susceptible to manipulation by adversaries, that can lead to errors in state estimation. This paper presents a general framework to study attacks on state estimation by adversaries capable of injecting bad-data into measurements and further, of jamming their reception. Through these two techniques, a novel `detectable jamming' attack is designed that changes the state estimation despite failing bad-data detection checks. Compared to commonly studied `hidden' data attacks, these attacks have lower costs and a wider feasible operating region. It is shown that the entire domain of jamming costs can be divided into two regions, with distinct graph-cut based formulations for the design of the optimal attack. The most significant insight arising from this result is that the adversarial capability to jam measurements changes the optimal 'detectable jamming' attack design only if the jamming cost is less than half the cost of bad-data injection. A polynomial time approximate algorithm for attack vector construction is developed and its efficacy in attack design is demonstrated through simulations on IEEE test systems.
\end{abstract}

\section{Introduction}
As power grids around the world move towards smarter devices and distributed control, it has led to large scale placement of cyber meters like PMUs \cite{pmu1} for real-time data collection. This can have a variety of positive implications for the grid, notably monitoring of the grid state for improved reliability and optimal electricity prices. However, `smart' meters and associated communication infrastructure are vulnerable to adversarial attacks by rogue agents and online viruses. Examples of these attacks include GPS spoofing attack on PMUs \cite{todd}, `Dragonfly' virus \cite{dragonfly}, Arora test attack \cite{arora} among others. Such data attacks can lead to incorrect estimation of the grid state and result to large scale blackouts. The extreme consequences of adversarial attacks and counter strategies has attracted significant interest from the research community. \cite{hidden} first introduced the problem of undetectable data attacks that bypass standard bad-data tests present in the state estimator. The optimal attack vector comprising of the compromised measurements is constructed in \cite{hidden} using projection matrices. Subsequent work has looked at the problem of constructing the optimal attack under different grid conditions and adversarial objectives. Attack construction that require minimum number of measurement corruptions are presented in \cite{poor} using $l_0 -l_1$ relaxation. Reference \cite{sou} analyzed a system with phasor measurements and used mixed integer linear programming to create the optimal attack. For systems with phasor and line flow measurements and PMUs, \cite{deka,deka1} discusses graph cut based attack designs on specific buses on the grid and associated protection strategies. Similarly, other protection schemes have been discussed in literature, including heuristic protection schemes \cite{thomas}, greedy schemes \cite{poor,deka1} among others.

It is worth noting that most research on power grid cyber-security has focussed on designing `hidden' attack vectors that completely evade the bad-data detection tests at the grid's state estimator. However, the authors of \cite{frame} showed that data `framing' attacks can be constructed that changes the values in half of the measurements in the attack vector while damaging the other half. The attack is initially detected by the estimator but becomes feasible after the bad-data identifier removes the damaged measurements. In \cite{dekaISGT}, a generalized `detectable' attack model was presented for systems where a subset of the measurements are incorruptible. The authors in \cite{dekaISGT} showed that by focussing on the bad-data identifier, the cardinality of the optimal `detectable' data attack in most cases can be reduced by greater than $50\%$ ($50\%$ in worst case) of that of `hidden' attacks. More importantly, the `detectable' attack framework in \cite{dekaISGT} is shown to produce feasible attacks in operating regimes that are secure against `hidden' attacks. In this work, we consider the `detectable' attack framework in \cite{dekaISGT} but with one major modification to the adversary's capability. In addition to modifying insecure measurements (bad-data injection) as described in previous work, the adversary considered here is capable of jamming or blocking measurement communication to the state estimator. Note that measurement jamming can be conducted using commercial jammers (for wireless communication), Denial of Service attack \cite{ddos} or by physically damaging the communication channel. Compared to bad-data injection that requires measurements to be changed by precise real values, measurement jamming is in fact less resource-intensive. One can make the realistic assumption that the non-negative cost of jamming lies in the range between $0$ and the cost of injecting bad-data into a measurement.

The overarching goal of this work is thus to \textit{study the impact of adding measurement jamming to the adversary's arsenal on the design of the optimal `detectable' data attacks.} Here, we formulate the optimal attack vector design as a graph cut problem based on the necessary and sufficient conditions for feasibility. We show that the entire range of values for measurement jamming cost can be divided into two intervals with different optimal attack formulations that lead to two distinct design strategies. Specifically, we prove that measurement jamming significantly alters the optimal `detectable' attack design only if the jamming cost is less than half the cost of data-injection. In contrast, we show that for `hidden' data attacks, measurement jamming leads to a single simple attack strategy independent of the jamming cost. We provide recursive min-cut based algorithms to design the optimal attack over the entire range of jamming cost values and show the cost improvement derived from measurement jamming through simulations on IEEE test cases \cite{testsystem}. By discussing the scope of measurement jamming as an adversarial strategy, our work thus provides a potent and realistic generalization of current data attack frameworks. Finally, we show that number of incorruptible measurements needed to prevent `detectable' attacks scales at least with the total number of measurements. This is much higher than `hidden' attacks where the security needs scale with the number of buses in the system \cite{deka1}. Thus, in addition to significantly reducing the cost of data attacks, our attack framework also undermines measures of grid resilience based on `hidden' attacks.

The rest of this paper is organized as follows. The next section presents a description of the system models used in state estimation, bad-data detection and identification. The novel adversarial attack model with jamming is introduced in Section \ref{sec:attack} along with conditions necessary for attack feasibility. Section \ref{sec:jammingattack} analyzes how the cost of jamming affects the attack strategy and grid resilience and presents a graph theoretic formulation for the optimal attack design. Our  algorithm to design an optimal attack vector is presented in Section \ref{sec:algo}. Simulations of the proposed algorithm for the range of jamming and bad-data injection costs on IEEE bus systems and comparisons with existing work are shown in Section \ref{sec:results}. Finally, concluding remarks and future directions of work are presented in Section \ref{sec:conclusion}.

\section{State Estimation and Bad-Data Detection in Power Grids}
\label{sec:estimation}
We denote the power grid by a set $V$ of buses (nodes) connected by a set $E$ of transmission lines (directed edges). Figure \ref{14bus} shows the graph representation of the IEEE $14$ bus test system \cite{testsystem}.
\begin{figure}
\centering
\includegraphics[width=0.43\textwidth, height = .28\textwidth]{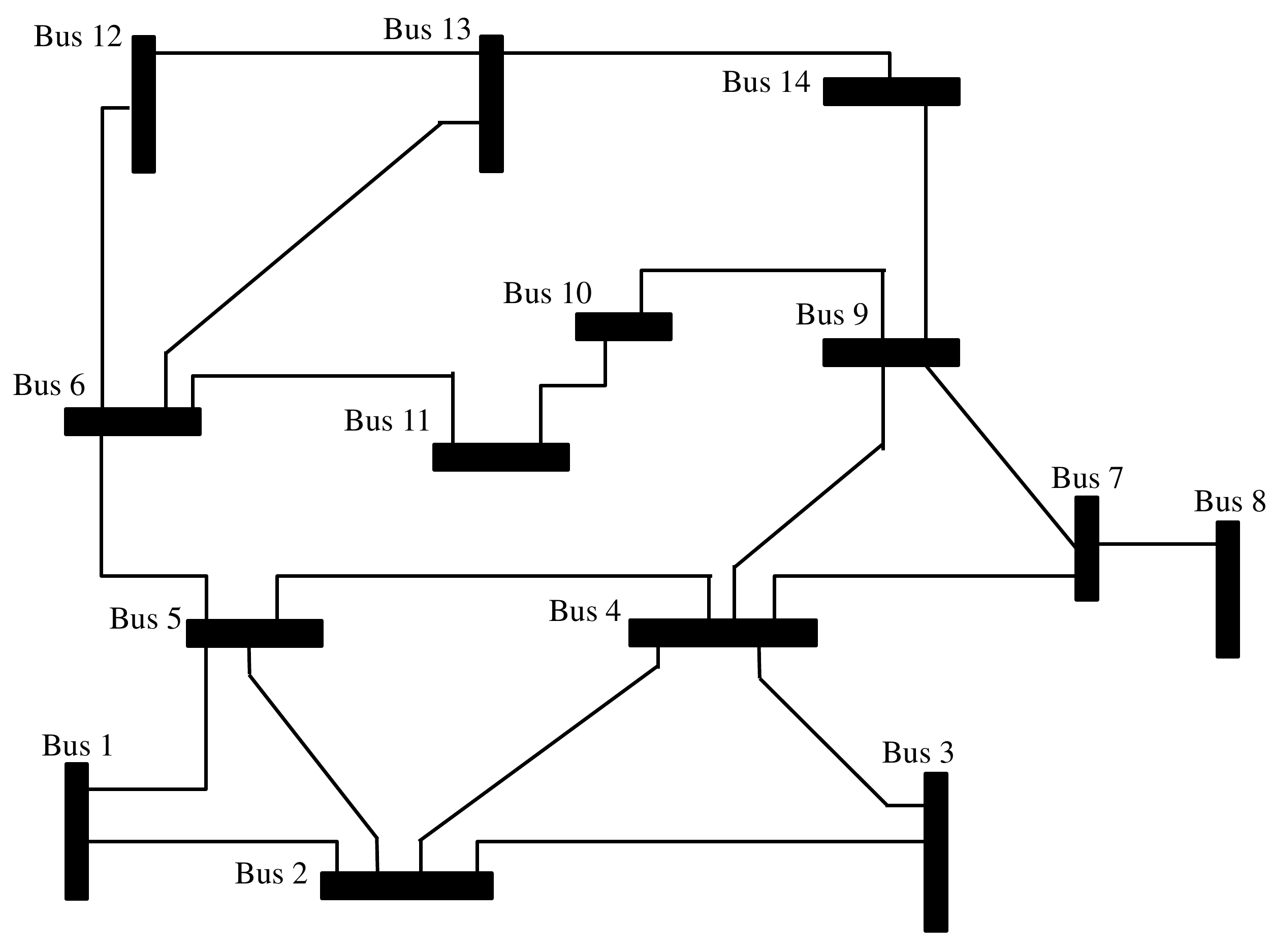}
\caption{IEEE 14-bus test system \cite{testsystem}}
\label{14bus}
\end{figure}

\textbf{Measurement Model:} We use DC power flow model \cite{abur} for the grid here where nodal line voltage magnitudes and line resistances are ignored. It is given by:
\begin{align}
z = Hx + e \label{dcmodel}
\end{align}
Here $z \in \mathbb{R}^m$ is the $m$ length vector of measurements. We consider two kinds of measurements in the grid: a) flow measurements on lines and b) voltage phasor measurements on buses, measured by conventional meters and phasor measurement units. $x \in \mathbb{R}^n$ denotes the state vector of length $n = |V|$ that comprises of the phase angles at the buses in the grid. $H$ is the measurement matrix and $e$ is a zero mean Gaussian measurement noise vector with known covariance $\Sigma$. Let the $k_1^{th}$ and $k_2^{th}$ entries in $z$ represent the power flow on line $(i,j)$ from nodes $i$ to $j$ and the voltage phasor at node $i$ respectively. Then, $z(k_1) = B_{ij}(x(i)-x(j)),~ z(k_2) = x(i)$. Here $B_{ij}$ is the susceptance of line $(i,j)$. The corresponding rows in $H$ thus have the following structure:
\begin{align}
H(k_1) = [0..0~~B_{ij}~~ 0..0~~-B_{ij}~~0..0] \label{line}\\
H(k_2) = [0..0~~ 1~~ 0..0] \label{bus}
\end{align}

We assume $m>n$ and full column rank of $H$, without a loss of generality. Further, without a loss of generality, we introduce a $(n+1)^{th}$ reference bus with phase angle $0$ in our system and represent it by augmenting $0$ to the state vector $x$. Let $z$ include the phase angle measurement for some bus $i$. Note that the angle measured can be considered equivalent to a flow on a hypothetical line of unit conductance between bus $i$ and the reference bus (with phase $0$). Thus, we can add an extra binary valued column $h^{g}$ corresponding to the reference bus in matrix $H$ to get $z = Hx = [H|h^g]\setlength{\arraycolsep}{2pt} \renewcommand{\arraystretch}{0.8}\begin{bmatrix} x \\0 \end{bmatrix}$. Here $h^g(k) = -1$ if $z(k)$ measures a phase angle and $0$ otherwise. Observe that after addition of the reference bus in the system, all measurements now correspond to flow measurements. Abusing notation, we use $x$ and $H$ to denoted the augmented state vector and measurement matrices respectively from this point.

\textbf{State Estimator:} We consider a least-square state estimator in the grid as shown in Figure \ref{estimator} \cite{monticelli,abur}.
\begin{figure}
\centering
\includegraphics[width=0.44\textwidth]{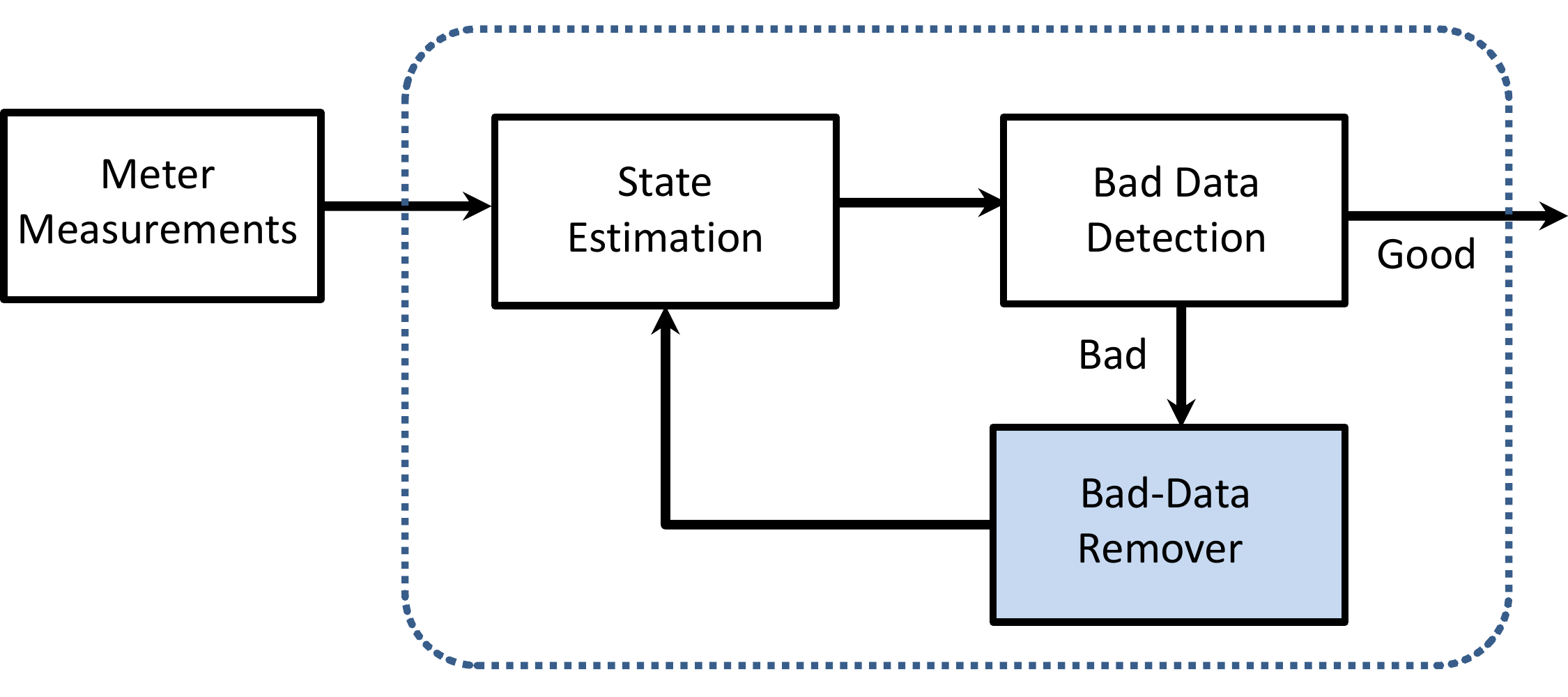}
\caption{State Estimator for a power system \cite{monticelli,abur}}
\label{estimator}
\end{figure}

The state vector estimate $x^*$ for a given measurement vector $z$ is generated by minimizing the weighted measurement residual $J(x,z) = \|\Sigma ^{-.5}(z-Hx)\|_2$ over variable $x$. Following estimation, a threshold ($\lambda$) based bad-data detector determines the presence of erroneous measurements by the following test:
\begin{align}
\|\Sigma ^{-.5}(z-Hx^*)\|_2 &\leq \lambda ~~\text{accept~~~} x^*\nonumber\\
                            &> \lambda ~~\text{detect bad-data} \label{test}
\end{align}

If the test detects bad-data, the measurements are sent for eliminating the bad-data as described below, following which the state estimate is recomputed.

\textbf{Bad-data Removal:} Note that the measurement residue vector $r$ for measurement $z$ and estimated $x^*$ is given by \cite{monticelli, abur}:
\begin{align}
r = z-Hx^* = [I - H(H^T\Sigma^{-1}H)^{-1}H^T\Sigma^{-1}]z
\end{align}
with variance $R_r$. Assuming that each measurement is independently affected by natural bad data, the state estimator removes the least number of erroneous measurements such that the resulting residual satisfies the threshold condition in Eq.~(\ref{test}) while preserving full column rank in $H$.  For a single removal, the optimal strategy is to remove the measurement with largest normalized residual \cite{monticelli}. However, for multiple bad-data entries, the optimal removal strategy is a non-convex problem \cite{monticelli, dekaISGT}.

We assume in the remainder of this paper that the measurement data $z$, in the absence of any adversarial manipulation, is reasonably clean and capable of producing the correct state estimate $x^*$ by passing the bad-data detection test.

\subsection{Attack Models}

Let $a$ denote the injected adversarial attack vector that is added to correct measurements in $z$ to generate the compromised measurement vector $z+a$. Traditional attack models have focussed on bypassing the bad-data detector by ensuring that the measurement residual in Eq.~(\ref{test}) remains unchanged following the injection of bad-data. Mathematically, this requires $a = Hc \neq 0$ for some $c \in \mathbb{R}^n$ as $\|\Sigma ^{-.5}(z-Hx^*)\|_2 = \|\Sigma ^{-.5}(z+a-H(x^*+c))\|_2$. Thus, a \textbf{`Hidden' Attack} results that produces an erroneous state vector $x^* +c$ \cite{hidden}. Next we describe `detectable' data attacks \cite{dekaISGT} that are the focus of this paper.

\textbf{`Detectable' Data Attack:} From the bad-data removal scheme described earlier, it is clear that an attack vector $a \neq 0$ will change the state estimate if removal of some other $k < \|a\|_0$ measurements (distinct from the attack vector) satisfies the bad-data detection test. For a nonzero $Hc$, consider the adversarial strategy that excludes (or does not corrupt) less than $50\%$ of the non-zero entries in $Hc$ from the attack vector $a$. Note that $a$ still gives a feasible `detectable' attack as the non-zero terms in $(Hc-a)$ are identified as bad-data instead of vector $a$. This happens as $\|a\|_0 > \|Hc-a\|_0$. In the next section, we formulate in detail the design of the optimal `detectable' data attack and the use it to analyze changes that arise due to the adversarial capability to jam measurements.

\section{`Detectable' Attack with Measurement Jamming}
\label{sec:attack}
In a general setting, few of the measurements in the grid may be incorruptible due to geographical isolation or encryption. We denote this set of measurements secure from adversarial corruption by $S$. Note that measurements in $S$ suffer from normal bad-data arising from measurement noise. The remaining insecure measurements belong to set $S^c$. The measurements included in the minimum cost `detectable' attack are given by non-zero terms in the optimal vector $d*$ in the following optimization problem \cite{dekaISGT}:

\begin{align} \label{opt_attacknew} \tag{P-1}
&\smashoperator[l]{\min_{d \in \{0,1\}^m, c \in \mathbb{R}^{n+1}}} \|d\|_{0} \\
\text{s.t. ~} &a = Hc, c \neq \textbf{0}, c(n+1) = 0\nonumber\\
&d(i) = 0 ~\forall i \in S_{m} ~~(\text{secure measurements}) \nonumber\\
& \|d\|_{0} > \|a\|_0/2 ~~(\text{for feasibility})\label{cond1} \\
& rank(DH) = n,~ diag(D) = \textbf{1} - (\textbf{1}-d)*a_{spty} \label{cond2}
\end{align}

Here, $a*b$ refers to the element-wise multiplication between vector $a$ and $b$, while $a_{spty}$ denotes the sparsity pattern in vector $a$. Condition (\ref{cond1}) ensures that the estimator removes measurement entries corresponding to non-zero terms in $(\textbf{1}- d)*a$ as bad-data, instead of the data injected in $d*a$. $D$ is a diagonal matrix whose diagonal entries are $0$ for removed data and $1$ otherwise. $DH$ is the measurements matrix after bad-data removal. Condition (\ref{cond2}) keeps it at full rank. The attack passes the bad-data detection test as it lies in the column space of $DH$. It is worth restating that as each row in augmented $H$ corresponds to a flow measurement, $H$ is equivalent to a susceptance weighted incidence matrix of a graph $G_H$ with $n+1$ nodes and edges given by rows in $H$. Due to this structure of $H$, it can be shown that \cite{deka,deka1,dekaISGT} the optimal attack $a^* = Hc$ corresponds to a $0-1$ binary valued nodal vector $c$. Further, the optimal attack strategy for Problem $\ref{opt_attacknew}$ doesn't change if $H$ is replaced by the un-weighted incidence matrix $A_H$ of graph $G_H$ ($A_H(i,j) = 1(\hat{H}(i,j) > 0) - 1(\hat{H}(i,j) < 0)$) as for a binary valued $c$, $A_Hc$ and $Hc$ have the same set of non-zero terms (identical sparsity pattern). Note that non-zero values in $A_Hc$ actually represents cut edges in graph $G_H$ between nodes marked $1$ and $0$. This leads to the following result (Theorem 2 in \cite{dekaISGT}) for optimal attack for Problem \ref{opt_attacknew}.

\begin{theorem}[{\cite[Theorem 2]{dekaISGT}}] \label{previous}
Let $C^*$ denote the minimum cardinality cut in $G_H$ with a minority of secure cut-edges ($|C^* \cap S| < |C^*|/2 $ ). An optimal `detectable' attack for Problem \ref{opt_attacknew} is given by any $\lfloor1+ |C^*|/2\rfloor$ cut-edges in $C^* \cap S^c$ (insecure cut edges).
\end{theorem}

We ignore the proof here for space constraints. Observe that if $d$ is restricted to an all-$1$ vector, Problem \ref{opt_attacknew} reduces to the problem of determining the optimal `hidden' attack. The optimal attack in that case is given by the minimum cardinality cut in $G_H$ that does not include any secure edge in $S^m$ \cite{deka,deka1}.

\textbf{`Detectable Jamming' Attack:} We now analyze an adversary with the capacity to jam insecure measurements in addition to manipulating their values by bad-data injection. Secure measurements are assumed to be Let $p_J$ and $p_I$ be the cost associated with jamming and bad-data injection into an insecure measurement in the grid respectively. We assume that $0 \leq p_I\leq p_I$ as the range of $p_J$ as jamming is less resource intensive than bad-data injection. This is a reasonable assumption as jamming can even be conducted by introducing garbage values through bad-data injection techniques. For ease of elucidation, we assume that the jamming and manipulation costs are uniform over all measurements in $S^c$, though all analysis follows immediately for variable costs as well. Consider a cut $C$ in graph $G_H$. Let $n^C_S$ and $n^C_{S^c}$ denote the number of secure and insecure edges in cut $C$ with $n^C_{S^c} > n^C_S$ as shown in Fig.~\ref{fig:feasibleattack}. By Theorem \ref{previous}, attack feasibility requires injection into $k^C$ ($k^C > |C|/2$) insecure edges at a cost of $p_Ik^C$. Instead, consider a different strategy where the adversary jams $k^C_J$ insecure measurements. As jammed measurements are not received and ignored by the control center, the cut-size effectively reduces to $|C|-k^C_J$. If the remaining $n^C_{S^c}-k^C_J$ insecure edges in the cut are greater in number than the $n^C_S$ secure edges, the adversary can still attack $k^C_I \geq 1 + \lfloor\frac{|C|-k^C_J}{2}\rfloor$ measurements and generate a feasible attack. As depicted in Fig.~\ref{fig:feasibleattack}, the cost of this new attack is $p_Ik^C_I+p_Jk^C_J$. We term it a `detectable jamming' attack to distinguish it from the original `detectable' attack that doesn't incorporate jamming.

We formulate the design of the optimal `detectable jamming' attack as follows:

\begin{align} \label{opt_attacknew1} \tag{P-2}
&\smashoperator[l]{\min_{d_J, d_I \in \{0,1\}^m}} p_J\|d_J\|_{0}+ p_I\|d_I\|_{0} \nonumber\\
\text{s.t. ~} &a = A_Hc, c \in \{0,1\}^{n+1}-{\textbf{0}}, c(n+1) = 0\nonumber\\
& d_J+d_I \in \{0,1\}^m \label{cond1j}\\
&d_J(i)= d_I(i) = 0 ~\forall i \in S_{m} \label{cond2j}\\
& \|d_I\|_{0} > (\|a\|_0- \|d_J\|_0)/2 ~~(\text{for feasibility})\label{cond3j} \\
& rank(DA_H) = n \text{~~where~} diag(D) = \textbf{1} -(\textbf{1}-d_J-d_I)*|a|\label{cond4j}
\end{align}

The non-zero values in optimal $d_J$ and $d_I$ give the measurements to jam and injection bad-data respectively in the optimal attack. Note that in Problem \ref{opt_attacknew1}, we replaced $H$ with incidence matrix $A_H$ and made $c$ a $0-1$ vector as discussed earlier. Here, condition \ref{cond1j} ensures data injection and jamming cannot occur at the same measurement. The remaining conditions arise from incorruptibility of secure measurements (\ref{cond2j}), feasibility of `detectable' attack (\ref{cond3j}) and full system observability after bad-data removal (\ref{cond4j}). From the discussion preceding Problem \ref{opt_attacknew1}, it is clear that the optimal `detectable jamming' attack has a graph-cut based construction as stated below.

\begin{lemma}\label{construction}
Let $C$ denote a cut in $G_H$ with $(n^C_{S^c}> |C|/2)$ insecure cut-edges. A feasible attack is given by jamming $(k^C_J \geq 0)$ and injecting data into $(\lfloor1+ (|C^*|- k^C_J)/2\rfloor > 0)$ of the $n^C_{S^c}$ insecure cut-edges at a cost of $p_Jk^C_J+p_I\lfloor1+ (|C^*|- k^C_J)/2\rfloor$. The optimal `detectable jamming' attack is given by minimizing the attack cost over variable $k^C_J$ (jammed edges) for all feasible cuts $C$.
\end{lemma}

It is noteworthy that if $k^C_J = 0$ in Lemma \ref{construction}, we obtain the optimal `detectable' attack (no jamming) as a feasible `detectable jamming' attack. This leads to following important properties.

\begin{corollary}\label{constcorollary}
\begin{itemize}
\item The space of system configurations with feasible `detectable jamming' attacks is identical to that of `detectable' attacks and is a superset of that of hidden attacks.
\item The cost of the optimal `detectable jamming' attack is never greater than the cost of optimal `detectable' attack and never greater than $.5+ 1/|C_h^*|$ times the cost of optimal `hidden' attack on a system, $|C_h*|$ being the cardinality of optimal `hidden' attack.
\end{itemize}
\end{corollary}

The first property arises as the set of cuts with majority of edges in $S^c$ (feasibility requirement of `detectable' and `detectable jamming' attacks) is a superset of the set of cuts will all edges in $S^c$ (feasibility requirement of `hidden' attacks). The second property has two parts: the first part follows from the fact that the optimal `detectable' attack is a feasible `detectable jamming' attack and hence not of lower cost that the optimal; the second part follows from the fact that injecting bad-data into $1 + \lfloor|C_h^*|\rfloor/2$ measurements of the optimal `hidden' attack constitutes a feasible `detectable' attack. It needs to be mentioned that these bounds reflect comparisons in the worst-case. The simulation results in Section \ref{sec:results} demonstrate that the average impact of `detectable jamming' attack is much more substantial. In the next section, we discuss the effect of jamming cost $p_J$ on the design of the optimal attack vector and its key properties.

\section{Effect of Jamming cost on Attack Construction}
\label{sec:jammingattack}
\begin{figure}[ht]
\centering
\includegraphics[width=0.50\textwidth]{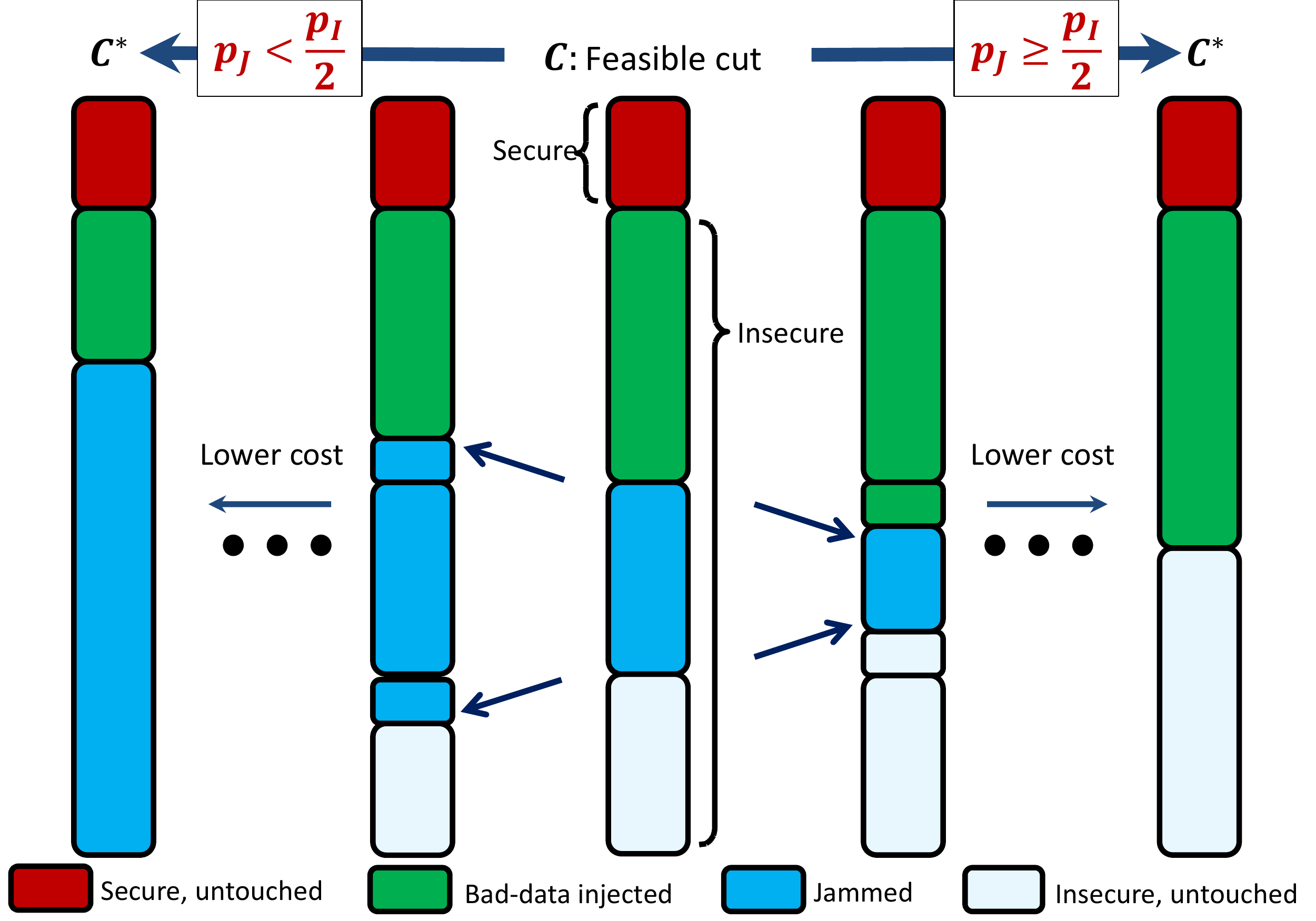}
\caption{Effect of jamming cost $p_J$ and bad-data injection cost $p_I$ on the minimum cost attack $C^*$ derived from a feasible cut $C$ with $n^C_S$ secure and $n^C_{S^c}$ insecure measurements. Secure, insecure but untouched, jammed, bad-data injected measurements in the cut are represented by red, white, blue and green colors respectively. When $p_J<p_I/2$, attack cost is reduced by replacing one bad-data injection with jamming two measurements as shown in the cuts on the left of $C$. For $p_J \geq p_I/2$, attack cost is reduced by replacing two jammed measurements by one measurement with bad-data injection while leaving the other untouched as shown on the right side of cut $C$. Optimal cuts $C^*$ got from this replacement are given by Theorem \ref{attackconstruction}.}
\label{fig:feasibleattack}
\end{figure}

As mentioned earlier, we consider the jamming cost $p_J$ to lie in the interval $[0, p_I]$ where $p_I$ is the bad-data injection cost. Consider a feasible cut $C$ with $n^C_{S^c}$ insecure edges and $n^C_S$ secure edges in the measurement graph $G_H$. Here $n^C_{S^c} > n^C_{S}$ as shown in Fig. \ref{fig:feasibleattack}. By Theorem \ref{construction}, a feasible `detectable jamming' attack comprises of selecting $(k^C_J \geq 0)$ and $(k^C_I =\lfloor1+ (|C|- k^C_J)/2\rfloor > 0)$ insecure edges for jamming and bad-data injection respectively, at a overall cost of $p^C$
\begin{align}
p^C &= p_Jk^C_J+p_I\lfloor1+ (|C|- k^C_J)/2\rfloor \nonumber\\
    &= (p_J-p_I/2)k^C_J+p_I\frac{|C|+2-(|C|- k^C_J)\mod 2}{2}\label{attackcost}
\end{align}
We divide the range of $p_J$ into two intervals: A ($p_J < p_I/2$) and B ($p_I/2 \leq p_J \leq p_I$). Note that in interval A, the cost $p^C$ is a decreasing function of $k^C_J$. Therefore, the minimum cost attack for feasible cut $C$ is obtained by jamming $n^C_{S^c}-n^C_{S}-1$ (the maximum permissible number of) insecure edges . The remaining $n^C_{S}+1$ insecure edges, greater than the number of secure edges by one, are injected with bad-data. The attack cost is given by
\begin{align}
p^C &= p_J(n^C_{S^c}-n^C_{S}-1)+p_I(n^C_{S}+1)\nonumber\\
    &= (p_I-p_J)n^C_{S}+ p_Jn^C_{S^c}+ (p_I-p_J)\label{costA}
\end{align}.
Ignoring constant $(p_I-p_J)$, this equals $C$'s cut-weight if secure and insecure edges are given weights of $(p_I-p_J)$ and $p_J$ respectively. Thus, if $p_J < p_I/2$, the optimal `detectable jamming' cut corresponds to the feasible cut $C^*$ with lowest cut-weight in $G_H$, where secure and insecure edges have weights of $(p_I-p_J)$ and $p_J$ respectively. Next consider interval B ($p_I/2 < p_J \leq p_I$). In Eq.~(\ref{attackcost}), if $k^C_J$ is reduced by $2$, the $(|C^*|- k^C_J)\mod 2$ term remains unchanged and the overall cost $p^C$ decreases. Hence the optical attack for cut $C$ corresponds to either $k^C_J = 0$ or $k^C_J = 1$, otherwise the attack cost can be reduced further. Checking the contribution of $(|C^*|- k^C_J)\mod 2$ term manually, we note that the optimal attack for cut $C$ is given by $(k^C_J = 0, k^C_I = (1+|C|)/2)$ for odd $|C|$, and $(k^C_J =1, k^C_I = |C|/2)$ for even $|C|$. 
In either case, the optimal attack cost is an increasing function of the cut-size $|C|$ expressed below.
\begin{align}
p^C &= p_J(1 - |C|\mod 2) + p_I\lfloor (1+|C|)/2 \rfloor \label{costB}
\end{align}
Thus, in interval B, the optimal `detectable jamming' attack corresponds to the feasible cut $C^*$ with lowest cut-size in $G_H$. We summarize this discussion by presenting our main theorem for optimal `detectable jamming' attack construction.

\begin{theorem}\label{attackconstruction}
The minimum cost `detectable jamming' attack for measurement graph $G_H$ with jamming cost $p_J$ and bad-data injection cost $p_I$ is constructed as follows.
\begin{itemize}
\item \textbf{$p_J < p_I/2$:} Give weights of $p_I-p_J$ and $p_J$ to secure and insecure edges respectively in $G_H$ and find the minimum weight feasible cut $C^*$ with $n^{C^*}_{S}$ secure edges. Use $(n^{C^*}_{S}+1)$ insecure measurements for bad-data injection and jam the rest. 
\item \textbf{$p_J \geq p_I/2$:} Find the minimum cardinality feasible cut $C^*$ in $G_H$. Use $\lfloor (1+|C^*|)/2 \rfloor$ insecure measurements for bad-data injection and jam $(1 - |C^*|\mod 2)$ measurement. 
\end{itemize}
\end{theorem}

The comparison of attack costs in `detectable jamming' attacks with that of standard `detectable' attacks is given by the following.
\begin{theorem}\label{cots improvement}
\begin{itemize}
Let cut $C^*$ with $n^{C^*}_S$ secure and $n^{C^*}_{S^c}$ insecure edges correspond to the optimal `detectable' attack (no jamming). The cost of optimal `detectable jamming' attack satisfies the following bounds.
\item For $p_J < p_I/2$, the cost of the optimal `detectable jamming' attack is less than that of the optimal `detectable' attack cost by at least $(p_I -2p_J)\lfloor\frac{n^{C^*}_{S^c}-n^{C^*}_S}{2}\rfloor + p_J (1 -|C^*| \mod 2)$.
\item For $p_J\geq p_I/2$, the cost of the optimal `detectable jamming' attack is less than that of `detectable' attack by $p_I-p_J$ (if $|C^*|$ is even), and equal otherwise.
\end{itemize}
\end{theorem}
\begin{proof}
For $p_J \geq p_I/2$, using Theorem \ref{previous} and Theorem \ref{attackconstruction}, it follows that the optimal cuts for `detectable' and `detectable jamming' attacks are identical. The difference is costs follows immediately from the attack construction using the optimal cut $C^*$ in either case. For $p_J < p_I/2$, note that the minimum-cost `detectable jamming' attack for feasible cut $C^*$ is given by injecting bad-data into $n^{C^*}_S+1$ edges and jamming the other insecure edges. The difference in cost between `detectable jamming' attack and `detectable' attack for cut $C^*$ is thus:
\begin{align}
\delta &= p^I\lfloor 1 +\frac{n^{C^*}_{S^c} + n^{C^*}_{S}}{2}\rfloor - p_I(n^{C^*}_{S}+1) - p_J(n^{C^*}_{S^c}- n^{C^*}_{S}-1)\nonumber\\
& = (p_I -2p_J)\lfloor\frac{n^{C^*}_{S^c}-n^{C^*}_S}{2}\rfloor + p_J (1 -|C^*| \mod 2)\label{bound}
\end{align}
As $C^*$ is a feasible `detectable jamming' attack (not necessarily optimal) in this case, Eq.~ \ref{bound} gives a lower bound on the difference in optimal costs.
\end{proof}
Further, the following statements holds:
\begin{corollary}\label{specialcases}
\begin{itemize}
\item For $p_J = 0$ (minimum jamming cost), the optimal `detectable jamming' attack corresponds to the cut $C^*$, which has the minimum number of secure edges among all feasible cuts.
\item For $p_J = 0$, if a `hidden' attack exists, an optimal `detectable jamming' attack corresponds to the same cut $C^*$.
\end{itemize}
\end{corollary}

Finally, the following theorem presents the potency of `detectable jamming' attacks by a lower bound on the number of secure measurements required for complete security.

\begin{theorem}\label{protection}
A system is always vulnerable to `detectable jamming' attacks if less than half the total number of measurements are secure.
\end{theorem}
\begin{proof}
Consider the graph $G_H$ generated from the measurement system. A feasible `detectable jamming' attack requires a cut in $G_H$ with a majority of insecure edges. As less than half of the measurements in $G_H$ are secure, there is at least one bus connected with a majority of insecure edges. Thus, a feasible `detectable jamming' attack can be constructed using that bus's edges as the cut. Hence proved.
\end{proof}

Note that Theorem \ref{protection} provides a $O(|E|)$ lower bound on the minimum number of secure measurements required for complete security, that scales with the total number of measurements. In contrast, complete protection from `hidden' attacks require a maximum of $|V|-1$ secure measurements  \cite{hidden, deka1}, that is much lesser that in general graphs. In Section \ref{sec:results}, we show simulations that confirm that `detectable jamming' attacks are more resilience to presence of secure measurements than `hidden' attacks. In the next Section, we present our algorithm to construct the optimal attack described in Theorem \ref{attackconstruction} and Corollary \ref{specialcases}.

\section{Algorithm For Attack Construction}
\label{sec:algo}
To confirm the existence of a feasible attack, we need to identify a feasible cut with a majority of insecure edges in the graph. Theorem $3$ in \cite{dekaISGT} proves that this is equivalent to the `ration-cut' problem, a known NP-hard problem. Thus, the design of the optimal `detectable jamming' attack, in the worst case, is hard as well.

We now provide an approximate algorithm (Algorithm $1$) for attack vector construction. For $p_J<p_I/2$ (interval A), we create weighted graph $G_H$ with secure (insecure) edges having weight $p_I-p_J$ ($p_J$). For $p_J \geq p_I/2$ (interval B), we consider unweighted $G_H$. Using Theorem \ref{attackconstruction}, the optimal attack, in either case, is given by the minimum weighted feasible cut in $G_H$.

\textbf{Working:} Algorithm $1$ proceeds by computing the minimum weight cut $C$ in $G_H$ (Step \ref{step1}) and checks if it is a feasible cut (Step \ref{step2}). If $C$ is infeasible, one secure edge is selected randomly in $C$ and its edge-weight is increased by $\beta$ (Step \ref{step3}). We consider two cases, one where $\beta$ is taken as the secure edge-weight and the other where it is taken as $\infty$. Following the increase, the algorithm recomputes the minimum weight cut and checks for feasibility. This process is iterated until a feasible cut is obtained (construct the attack vector) or the cut-weight reaches a threshold $\gamma<\infty$(declare no solution).
\begin{algorithm}
\caption{`Detectable Jamming' Attack Construction}
\textbf{Input:} Graph $G_H$ with secure and insecure edges weighted based on $p_J, p_I$, $S, S^c, \beta,\gamma$ \\
\begin{algorithmic}[1]
\STATE Compute min-weight cut $C$ in $G_H$ \label{step1}
\STATE $w_C \gets $ weight of $C$
\WHILE {($w_C < \gamma, 2|C \bigcap S| \geq |C|$)} \label{step2}
\STATE Randomly pick edge $i \in C \bigcap S$ and increase its weight by $\beta$ \label{step3}
\STATE Compute min-weight cut $C$ in $G_H$
\STATE $w_C \gets $ weight of $C$
\ENDWHILE
\IF {$2|C \bigcap S| < |C|$}
\STATE Construct attack vector using Theorem \ref{attackconstruction}
\ELSE
\STATE Declare no solution
\ENDIF
\end{algorithmic}
\end{algorithm}

Note that for $\beta = \infty$, in the worst case, there are $|S|$ min-cut computations (one for each secure edge) of complexity $O(|V||E|+|V|^2\log|V|)$ giving the algorithm a computational complexity of $O(|S||V||E|+|S||V|^2\log|V|)$. However, as the algorithm is approximate, it might not return a solution in every case. In the next section, we show simulation results on designing optimal attacks by Algorithm $1$ in IEEE test systems. We also demonstrate the capacity of `detectable jamming' attacks in overcoming high placement of secure measurements in the systems considered.

\section{Results on IEEE test systems}
\label{sec:results}
We discuss the performance of Algorithm $1$ in designing `detectable jamming' attacks by simulations on IEEE $14$-bus and $57$-bus test systems \cite{testsystem}. In each simulation run, we put flow measurements on all lines in the test system considered and phase angle measurements on $60\%$ (randomly selected) of the system buses. Over multiple simulations, we vary the fraction of secure measurements and record the trends in average cost of constructing `detectable jamming' attack. We consider either interval of jamming cost ($p_J=0$, $p_J < p_I/2$ and $p_J > p_I/2$), and different values of parameter $\beta$ (finite and $\infty$) in Algorithm $1$. The trends in average optimal cost of `detectable jamming' attacks for the $14$ bus system are presented in Fig.~\ref{fig:14buspossible} (for configurations that allow feasible `hidden' attacks), and Fig.~\ref{fig:14busno} (for configurations that are resilient to `hidden' attacks). To demonstrate the efficacy of our approach, we compare the trends with average costs of constructing `hidden' and `detectable' (no jamming) attacks. Note that while the average attack cost is way below the upper bound  (Corollary \ref{constcorollary}) in Fig.~\ref{fig:14buspossible}, it is observed to eventually decrease with increasing secure measurements in the system. This trend results from the fact that system configurations resilient to attacks that increase with increasing secure measurements are not accounted for in the plotted average attack costs. Further, it is apparent that changing the value of $\beta$ does not affect the performance of Algorithm $1$ much. Similarly, Fig.~\ref{fig:57bus} includes the average cost trends for the $57$ bus system, with $\beta$ in Algorithm $1$ being taken equal to the weight of secure measurements. From the figures it is clear that jamming enabled attacks have significantly reduced costs over both `hidden' and `detectable' attacks. Finally, Fig.~\ref{fig:secure} plots the increase in number of completely resilient operating regimes (no feasible attack possible) with an increase in the number of secure measurements in the system. It is easily evidenced in Fig.~\ref{fig:secure} that compared to `hidden' attacks, `detectable' and `detectable jamming' attacks pose a much greater threat to the grid vulnerability as the number of secure operating regimes in the latter hardly increases with an increase in the number of secure measurements. This is in line with the security needs highlighted in Theorem \ref{protection}.
The simulations prove the dual adversarial benefits created by `detectable jamming' attacks: lowering of attack cost and increased insensitivity to deployment of incorruptible measurements.

\begin{figure}[ht]
\centering
\includegraphics[width=0.5\textwidth]{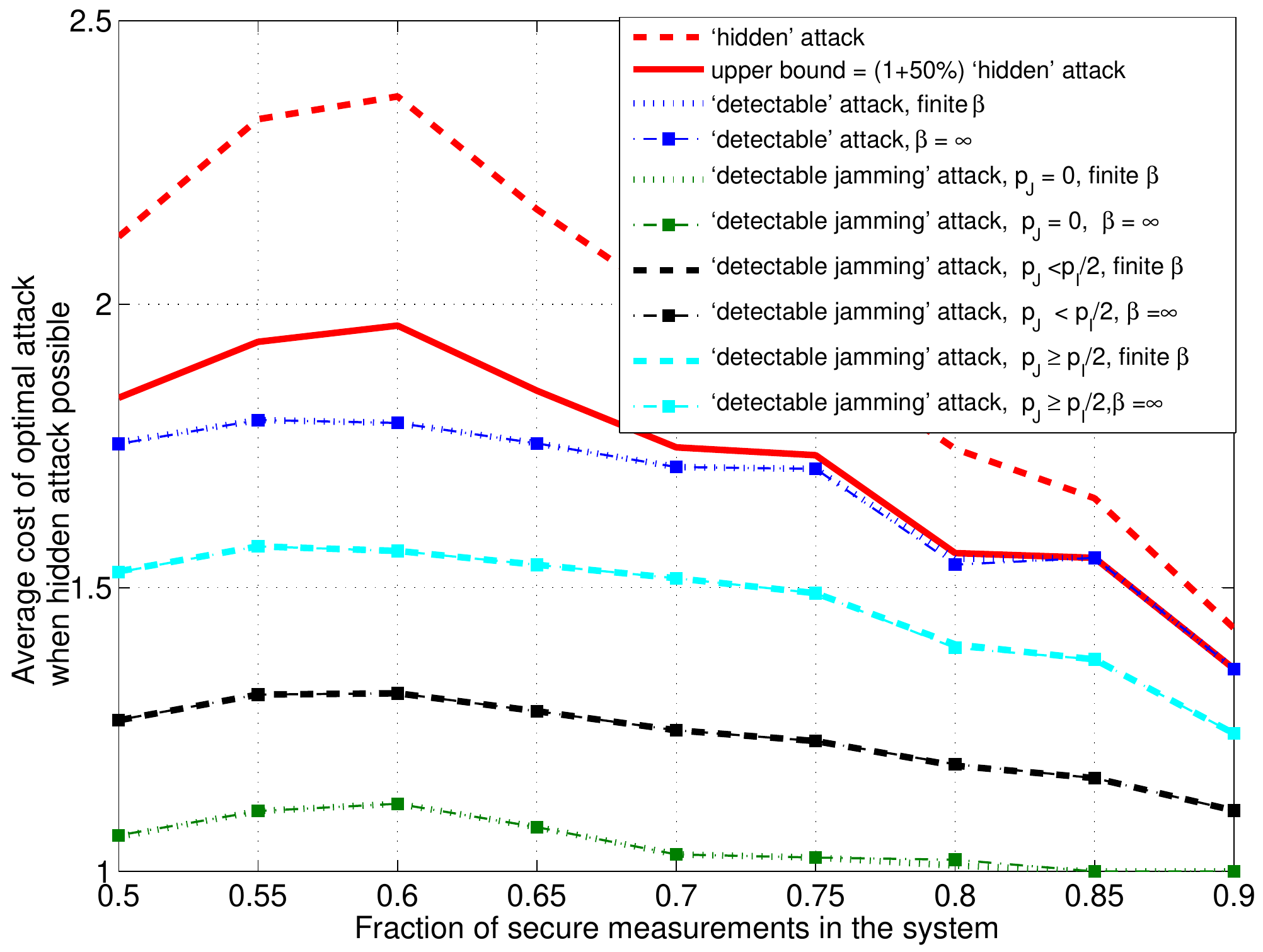}
\caption{Average cost of optimal attacks (`hidden', `detectable' and `detectable jamming') produced for different values of $\beta$ (size of secure edge and $\infty$) by Algorithm $1$ on the IEEE $14$ bus test system with flow measurements on all lines, phasor measurements on $60\%$ of the buses and protection on a fraction of measurements selected randomly. The bad-data injection cost ($p_I$) is taken as $1$. For the `detectable jamming' attack, the jamming costs ($p_J$) considered are $0, 1/4 (< p_I/2), 3/4 (> p_I/2)$. Only configurations where `hidden' attacks are possibly are considered to determine the average costs.}
\label{fig:14buspossible}
\end{figure}
\begin{figure}
\centering
\includegraphics[width=0.485\textwidth]{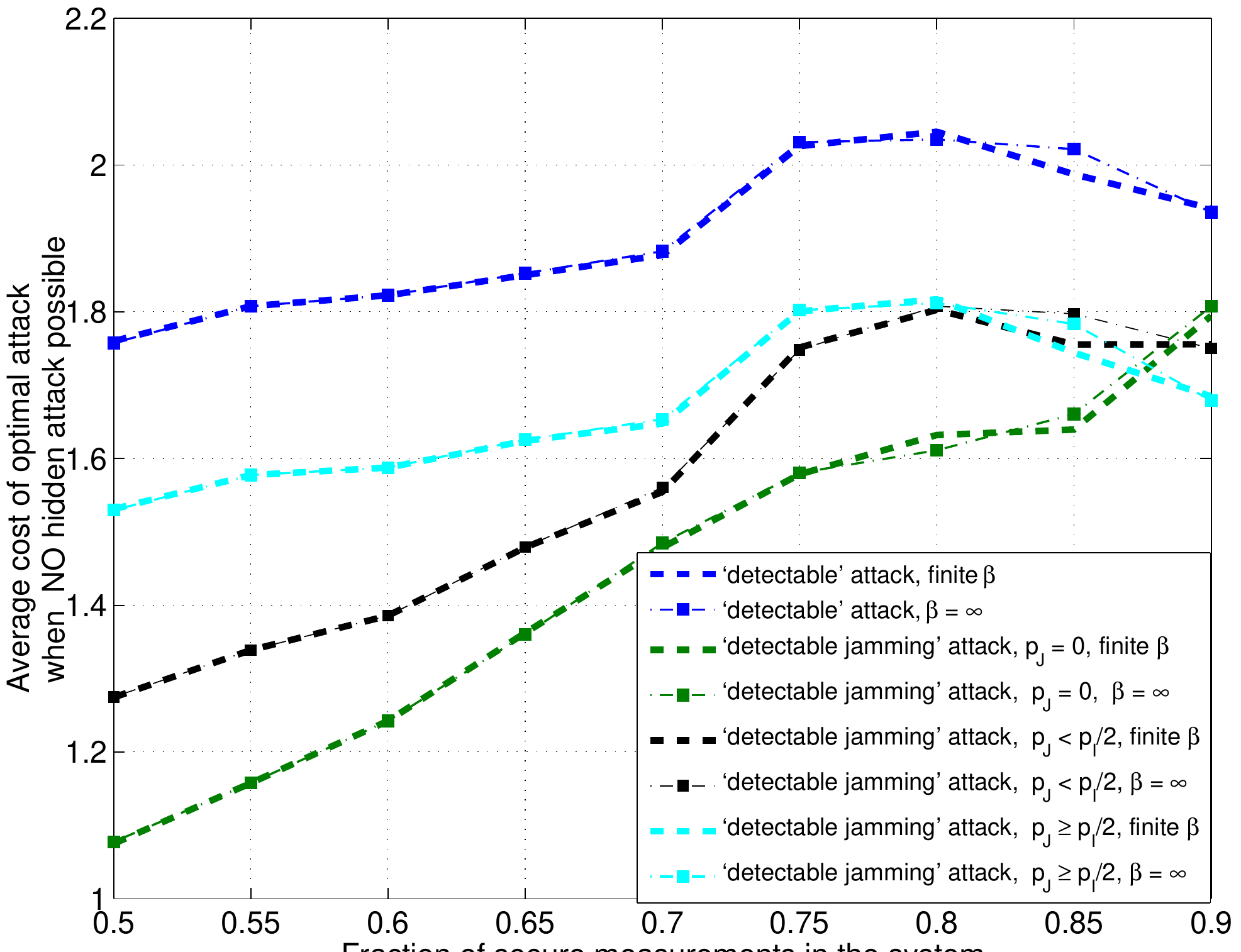}
\caption{Average cost of optimal attacks (`detectable' and `detectable jamming') produced for different values of $\beta$ (size of secure edge and $\infty$) by Algorithm $1$ on the IEEE $14$ bus test system with flow measurements on all lines, phasor measurements on $60\%$ of the buses and protection on a fraction of measurements selected randomly. The bad-data injection cost ($p_I$) is taken as $1$. For the `detectable jamming' attack, the jamming costs ($p_J$) considered are $0, 1/4 (<p_I/2), 3/4 (>p_I/2)$. Only configurations which are resilient against `hidden' attacks are considered to determine the average costs.}
\label{fig:14busno}
\end{figure}
\begin{figure}
\centering
\includegraphics[width=0.5\textwidth]{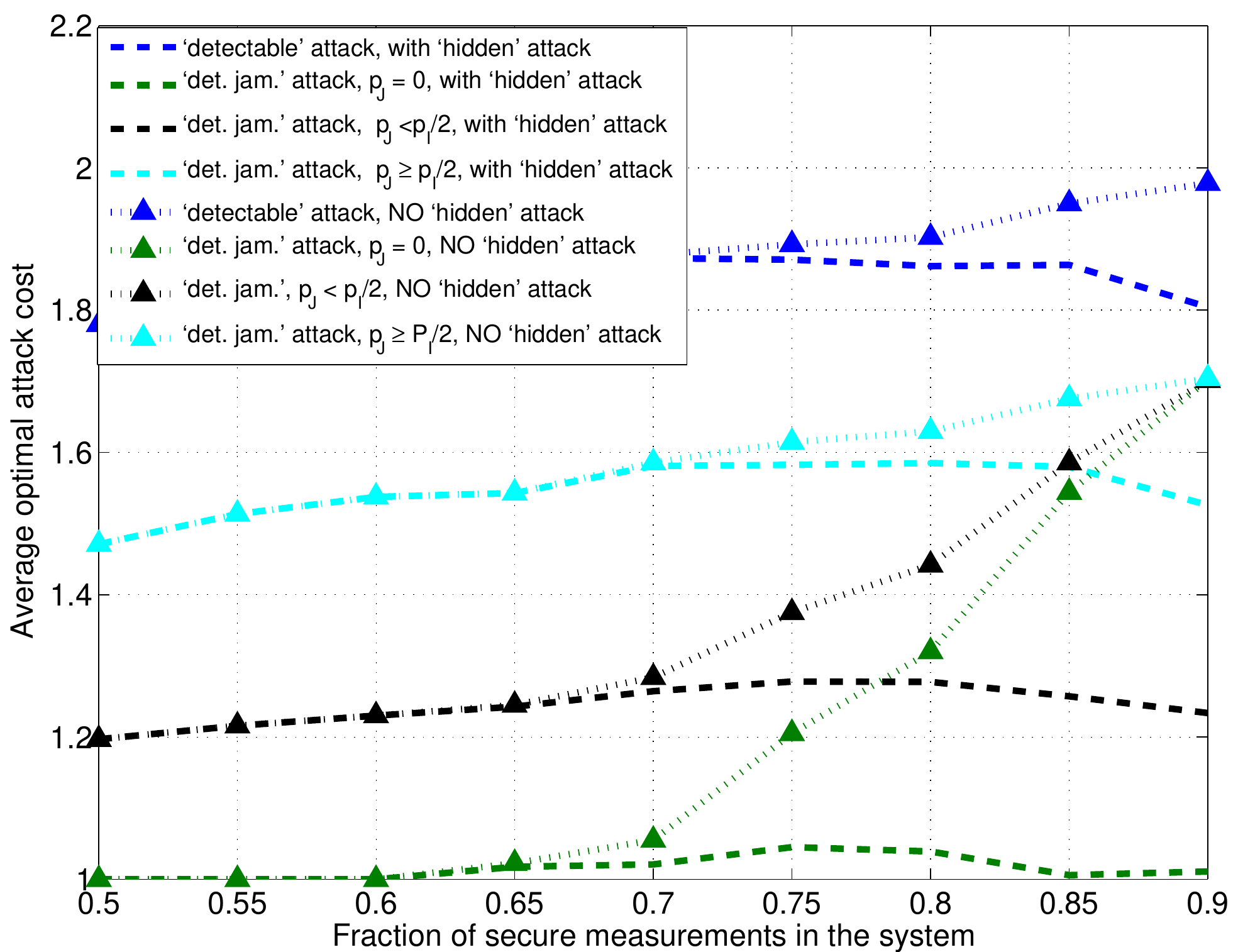}
\caption{Average cost of optimal attacks (`detectable' and `detectable jamming') produced by Algorithm $1$ (with finite $\beta$) on the IEEE $57$ bus test system with flow measurements on all lines, phasor measurements on $60\%$ of the buses and protection on a fraction of measurements selected randomly. The bad-data injection cost ($p_I$) is taken as $1$. For the `detectable jamming' attack, the jamming costs ($p_J$) considered are $0, 1/4 (<p_I/2), 3/4 (>p_I/2)$.}
\label{fig:57bus}
\end{figure}
\begin{figure}
\centering
\includegraphics[width=0.5 \textwidth]{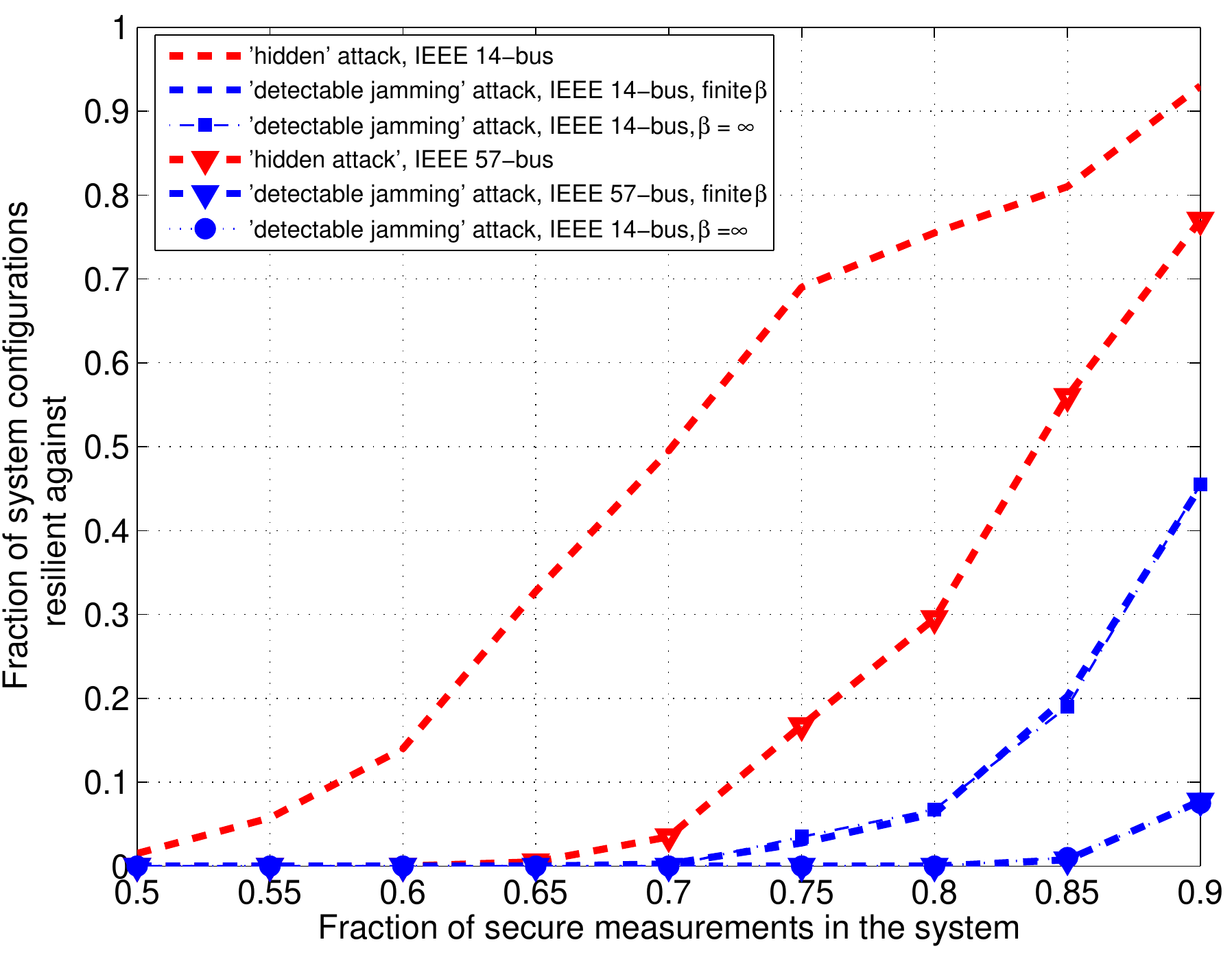}
\caption{Average fraction of simulated configurations with no feasible `hidden' and `detectable jamming' attacks given by Algorithm $1$ for different values of $\beta$ in IEEE $14$ and $57$ bus test systems. Each test system has flow measurements on all lines, phasor measurements on $60\%$ of the buses and protection on a fraction of measurements selected randomly.}
\label{fig:secure}
\end{figure}

\section{Conclusion}
\label{sec:conclusion}
We introduce a new data attack framework on power grids termed `detectable jamming' attacks, where an adversary uses measurement jamming as a tool in addition to changing meter readings (bad-data injection). Through the use of these dual techniques on an optimal set of measurements, the adversary creates a violation of the bad-data detection test but still creates a change in the estimated state vector. This is ensured by leading the state estimator to incorrectly label uncorrupted correct data as bad-data. We show that the design of the minimum cost attack of this regime is shown to be equivalent to a constrained graph cut problem that takes two different forms, dependent on the relative values of jamming and data injection costs. We prove that even the worst-case attack cost of `detectable jamming' attacks is approximately half of the optimal `hidden' attack cost, while the capability to overcome incorruptible measurements is much more pronounced than in `hidden' attacks. This is highlighted by the fact that the number of secure measurements required for complete resilience against `hidden' attack is of the order of number of buses in the system, while complete resilience against `hidden' attacks requires greater than half the measurements to be incorruptible and scales with the number of edges in the measurement graph. We further show that in comparison to `detectable' (no jamming) attacks, our jamming reliant framework significantly alters the optimal attack (given by the optimal graph cut) only if the jamming cost is less than half the cost of bad-data injection. For values of jamming cost greater than half the injection cost, `detectable jamming' attacks have a lower attack cost but correspond to the same optimal graph cut as `detectable' attacks. As the design of the optimal attack is NP hard in general, we present an iterative min-cut based approximate algorithm with polynomial complexity to determine the optimal cut. We demonstrate the adversarial benefits of our proposed attack framework and performance of our approximate algorithms through simulations on IEEE test cases for different values of jamming costs and different system conditions. This paper exposes the adverse effects to grid security posed by measurement jamming when used as an adversarial tool to supplement `bad-data' injection. Designing optimal security measures against this attack regime is the object of our current research in this domain.

\end{document}